\documentclass{eptcs-modified}

\usepackage{amsmath}
\usepackage{amssymb}
\usepackage{url}
\usepackage{amsthm}
\usepackage{graphicx}
\usepackage{amscd}

\usepackage{pgf} % added by SLR
\usepackage{tikz} % added by SLR

\begin{document}

\theoremstyle{definition}
\newtheorem{theorem}{Theorem}
\newtheorem{definition}[theorem]{Definition}
\newtheorem{problem}[theorem]{Problem}
\newtheorem{assumption}[theorem]{Assumption}
\newtheorem{corollary}[theorem]{Corollary}
\newtheorem{proposition}[theorem]{Proposition}
\newtheorem{example}[theorem]{Example}
\newtheorem{lemma}[theorem]{Lemma}
\newtheorem{observation}[theorem]{Observation}
\newtheorem{fact}[theorem]{Fact}
\newtheorem{question}[theorem]{Open Question}
\newtheorem{conjecture}[theorem]{Conjecture}
\newtheorem{addendum}[theorem]{Addendum}
\newcommand{\uint}{{[0, 1]}}
\newcommand{\Cantor}{{\{0,1\}^\mathbb{N}}}
\newcommand{\name}[1]{\textsc{#1}}
\newcommand{\id}{\textrm{id}}
\newcommand{\dom}{\operatorname{dom}}
\newcommand{\hide}[1]{}

\title{Equilibria in multi-player multi-outcome infinite sequential games}
%\subtitle{Subtitle Text, if any}

\author{
St\'ephane Le Roux
\institute{D\'epartement d'informatique\\ Universit\'e libre de
Bruxelles\thanks{The project was begun while the first author was a postdoctoral researcher at the University of Darmstadt, Germany.}, Belgium}
\email{Stephane.Le.Roux@ulb.ac.be}
\and
Arno Pauly
\institute{D\'epartement d'informatique\\ Universit\'e libre de
Bruxelles, Belgium\\ \& \\ Computer Laboratory\\ University of Cambridge, United Kingdom}
\email{Arno.Pauly@cl.cam.ac.uk}
}

\def\titlerunning{Multi-player multi-outcome infinite sequential games}
\def\authorrunning{S. Le Roux \& A. Pauly}
\maketitle

\begin{abstract}
We investigate the existence of certain types of equilibria (Nash, $\varepsilon$-Nash, subgame perfect, $\varepsilon$-subgame perfect, Pareto-optimal) in multi-player multi-outcome infinite sequential games. We use two fundamental approaches: one requires strong topological restrictions on the games, but produces very strong existence results. The other merely requires some very basic determinacy properties to still obtain some existence results. Both results are transfer results: starting with the existence of some equilibria for a small class of games, they allow us to conclude the existence of some type of equilibria for a larger class.

To make the abstract results more concrete, we investigate as a special case infinite sequential games with real-valued payoff functions. Depending on the class of payoff functions (continuous, upper semi-continuous, Borel) and whether the game is zero-sum, we obtain various existence results for equilibria.

Our results hold for games with two or up to countably many players.
\end{abstract}

\section{Introduction}
The present article continues the research programme to investigate sequential games in a very general setting, which was initiated by the first author in \cite{leroux3,leroux4}. It extends the conference paper \cite{paulyleroux2}. This programme reunites two mostly separate developments in the study of games. On the one hand, the first development is the investigation of variations on solution concepts for games, and of different formalizations of the preferences of the players, which primarily happened inside game theory proper. Related to that, game theory has also seen an interest in relaxing the continuity and convexity assumptions of \name{Nash}'s original existence theorem \cite{nash}. Both for an example and references to further work see \cite{nessah}. So typically, the countably infinite is absent from game theory: sets are either finite, or in cases such as randomized strategies, have the structure of the continuum.

On the other hand, the study of infinite sequential games has a long history in logic. Many variations on the rules of games have been studied, albeit mostly restricted to zero-sum games with two players and two outcomes. Thus, here the continuum is entirely absent (disregarding its internal occurrence in the set of potential plays), and the countably infinite is only used for time, not for e.g.~the number of agents or outcomes.

In our work we study infinite sequential games with perfect information (i.e.~generalized Gale-Stewart games \cite{gale2}) in a setting as general as possible. We can have any countable number of players, and we investigate various ways to represent preferences. Some results do put restrictions on the number of distinguished outcomes (as being countable), and some require a (generalized) zero-sum condition. A similar synthesis of the approaches is found in \cite{ummels,ummels2,flesch,purves}.

A classical result by \name{Martin} \cite{martin} established that such a game played by two players who only care about whether or not the play falls into some fixed Borel set is determined, i.e.~admits a winning strategy for one of the players. While determinacy can be understood as a special case of existence of Nash equilibrium, Nash equilibrium is often regarded as an unsatisfactory solution concept for sequential games in game theory. Subgame perfect equilibria are a more convincing solution concept from a rationality perspective. This article studies these two concepts, and as best responses are not always available, we also investigate the existence of $\varepsilon$-Nash equilibria and $\varepsilon$-subgame perfect equilibria. Furthermore, we also study Pareto-optimal equilibria.

The proofs that we provide fall into two broad categories: some of our results are obtained by lifting Borel determinacy to more complicated settings, similar to \cite{leroux3,leroux4} or to a sketched observation by \name{Mertens} and \name{Neyman} in \cite{mertens}. Other results are based on topological arguments to show strong existence results, albeit at the cost of continuity requirements in the game characterizations. Both proof techniques provide general results in rather abstract settings. In this sense, our main results are Theorem \ref{theo:cont} on the one hand, and Theorems \ref{thm:zs-fc-spe}, \ref{thm:zs-fo-spe}, \ref{thm:iac-ne} (which share parts of their proofs) on the other hand.

A common theme of the results is that they are transfer principles: they tell us how to take a pre-existing result on a more restricted class of games and obtain from it a result for a more general class of games. Theorem \ref{theo:cont} is then used to transfer the existence of subgame perfect equilibria from finite sequential games to certain infinite ones. With Theorems \ref{thm:zs-fc-spe}, \ref{thm:zs-fo-spe}, \ref{thm:iac-ne}, we can in particular extend Borel determinacy to yield equilibria in multi-player multi-outcome settings.

In an attempt to make the rather abstract results somewhat more accessible, we shall consider in addition the corollaries obtained in the situation where the goals of the players are to maximize real-valued payoff functions. Distinguishing properties here are continuity, upper or lower semicontinuity and Borel measurability. These settings have been studied before \cite{fudenberg,mertens,purves,flesch}, and usually our corollaries improve upon known results by extending them from the case of finitely many players to the case of countably many players. An overview of past and new results is given in the table on Page \pageref{fig:overview}.

Our results showcase which requirements are actually needed for which aspects of determinacy, and as such may contribute to the understanding of strategic behaviour in general.

Additionally, the emergence of quantitative objectives in addition to qualitative structure in traditional verification/synthesis games \cite{chatterjee2}, such as mean-payoff parity games \cite{chatterjee3,chatterjee4}, provides an area of applications for abstract theorems about the existence of equilibria. Existence results in such settings are usually not trivial, but are proven together with the introduction of the setting -- thus we do not answer open questions, but are hopeful that our results may be useful in the future. Having results for countably many players is important for applications if multi-agent interactions in open systems are studied. In order to employ an equilibrium existence result for finitely many players, a bound on the number of agents involved in the interaction might need to be common knowledge from the beginning on. Our results on the other hand easily enable a setting where additional agents may join the interaction later on, and only the number of agents who have acted in the past is finite.

This article is an extended and improved version of \cite{paulyleroux2}.

\begin{figure*}
\center
{\bf Overview}

\begin{tabular}{ll|cccc}
Payoff functions & Type & $\varepsilon$-Nash & Nash & $\varepsilon$-subgame perfect & subgame perfect\\
\hline
\hline
continuous & finitely many & yes & yes & yes & yes\\
& players & - & - & - & \cite[Corollary 4.2]{fudenberg} \\
\hline
continuous & countably many & yes & yes & yes & yes\\
& players & - & - & - & Corollary \ref{corr:contseque} \\
\hline
upper & finitely many & yes & yes & yes & no\\
semi-continuous & players & - & Corollary \ref{corr:uppercont} & \cite[Theorem 2.1]{purves} & Example \ref{ex:nosubgameperfect2}\\
\hline
upper & countably many & yes & yes & ? & no\\
semi-continuous & players & - & Corollary \ref{corr:uppercont} &  & Example \ref{ex:nosubgameperfect2}\\
\hline
lower & finitely many & yes & no & yes & no\\
semi-continuous & players & - & Example \ref{example:nonash} & \cite[Theorem 2.3]{flesch} & -\\
\hline
lower & countably many & yes & no & no & no\\
semi-continuous & players & - & Example \ref{example:nonash} & \cite[Subsection 4.3]{flesch} & -\\
\hline
Borel & zero-sum & yes & no & yes & no\\
& & - & Example \ref{example:nonash} & Corollary \ref{corr:esubgame} & - \\
\hline
Borel & finitely many & yes & no & no & no\\
& players & Corollary \ref{corr:borelhasenash} & Example \ref{example:nonash}& Example \ref{ex:nosubgameperfect} & - \\
\hline
Borel & countably many & yes & no & no & no\\
& players & Corollary \ref{corr:borelhasenash} & Example \ref{example:nonash}& Example \ref{ex:nosubgameperfect} & - \\
\end{tabular}

{\footnotesize Note: the table shows for any relevant combination of properties of the payoff functions, type of the game and type of equilibria whether such equilibria exist always or whether there is a counterexample. The results as listed here pertain to having finitely many choices at each stage of the game. There is no difference between two-player games and games with finitely many players in any situation we investigate. As we consider perfect information games only, any zero-sum game is understood to be a two-player game. The combination of semi-continuity and zero-sum would imply continuity, and is thus left out. For continuous payoff functions, we already see complete positive results without the zero-sum condition, and thus do not mention it explicitly either. Both Corollary \ref{corr:borelhasenash} and Example \ref{example:nonash} seem to be folklore results.}
\label{fig:overview}
\end{figure*}

\section{Background}
In our most abstract definition, a game is a tuple $\langle A, (S_a)_{a \in A}, (\prec_a)_{a \in A}\rangle$ consisting of a non-empty set $A$ of \emph{agents} or \emph{players}, for each agent $a \in A$ a non-empty set $S_a$ of \emph{strategies}, and for each agent $a \in A$ a \emph{preference} relation $\mathalpha{\prec}_a \subseteq \left (\prod_{a \in A} S_a \right ) \times \left (\prod_{a \in A} S_a \right )$. The generic setting suffices to introduce the notion of a Nash equilibrium: a \emph{strategy profile} $\sigma \in \left (\prod_{a \in A} S_a \right )$ is called a Nash equilibrium, if for every agent $a \in A$ and every strategy $s_a \in S_a$ we find $\neg \left (\sigma \prec_a \sigma_{a\mapsto s_a} \right )$, where $\sigma_{a\mapsto s_a}$ is defined by $\sigma_{a\mapsto s_a}(b) = \sigma(b)$ for $b \in A \setminus \{a\}$ and $\sigma_{a\mapsto s_a}(a) = s_a$. In words, no agent prefers over a Nash equilibrium some other situation that only differs in her choice of strategy.

We will give additional structure to games in two primary ways: in Section \ref{sec:continuityargument} we add topologies to the strategy spaces, and then impose some topological constraints on both strategy spaces and preferences. Beyond that, we will consider games where strategy spaces and preferences are derived objects from more structured variants of games. One such variant is the infinite sequential game:
\begin{definition}[Infinite sequential game, {cf.~\cite[Definition 1.1]{leroux3}}]\label{defn:ifg} An
  \emph{infinite sequential game}  is an object $\langle A,C,d,O,v,(\prec_a)_{a\in A}\rangle$ complying with the following.
\begin{enumerate}
\item $A$ is a non-empty set (of agents).
\item $C$ is a non-empty set (of choices).
\item $d:C^*\to A$ (assigns a decision maker to each stage of the game).
\item $O$ is a non-empty set (of possible outcomes of the game).
\item $v:C^\omega\to O$ (assigns outcomes to infinite sequences of choices).
\item Each $\prec_a$ is a binary relation over $O$ (modeling the preference of agent $a$).
\end{enumerate}
\end{definition}

The intuition behind the definition is that agents take turns to make a choice. Whose turn it is depends on the past choices via the function $d$. Over time, the agents thus jointly generate some infinite sequence, which is mapped by $v$ to the outcome of the game. Note that using a single set of actions $C$ for each step just simplifies the notation, a generalization to varying action sets is straightforward.

The infinite sequential games can be seen as abstract games: the agents remain the agents and the strategies of agent $a$ are the functions $s_a : d^{-1}(\{a\}) \to C$. We can then safely regard a strategy profile as a function $\sigma : C^* \to C$ whose \emph{induced play} is defined below, where for an infinite sequence $p \in C^\omega$ we let $p_n$ be its $n$-th value, and $p_{\leq n} = p_{<n + 1} \in C^*$ be its finite prefix of length $n$.
 \begin{definition}[Induced play and outcome, {cf.~\cite[Definition 1.3]{leroux3}}]\label{defn:ipo}
Let $s:C^*\to C$ be a strategy profile. The \emph{play $p=p^{\gamma}(s)\in
  C^\omega$  induced by $s$ starting at $\gamma \in C^*$} is defined inductively through its prefixes:
$p_n = \gamma_n$ for $n \leq |\gamma|$ and $p_{n}:=s(p_{< n})$ for $n > |\gamma|$. Also, $v\circ p^{\gamma}(s)$ is the \emph{outcome induced by $s$ starting at $\gamma$}. The play (resp. outcome) induced by $s$ is the play (resp. outcome) induced by $s$ starting at $\varepsilon$.
\end{definition}

In the usual way to regard an infinite sequential game as a special abstract game, an agent prefers a strategy profile $\sigma$ to $\sigma'$, iff he prefers the outcome induced by $\sigma$ to the outcome induced by $\sigma'$. And indeed we shall call a strategy profile of an infinite sequential game a Nash equilibrium, iff it is a Nash equilibrium with these preferences. In a certain notation overload, we will in particular use the same symbols for the preferences over strategy profiles and the preferences over outcomes.

However, there is a certain criticism of this choice as being not rational: essentially, the resulting concept of a Nash equilibrium means that players can use empty threats -- declarations they would play in a certain way from a position onwards, even if that would be against their own interests once the position is reached, as long as this threat keeps other players from moving to that position. This can be fixed by considering subgame perfect equilibria \cite{selten}. We understand these as the Nash equilibria derived from a different translation of preferences from infinite sequential games to abstract games. (Similar remarks were made in \cite[Lemma 144 in Section 7.2.3, Section 7.3.2]{SLR-PhD08}.)

\begin{definition}
Given an infinite sequential game $\langle A,C,d,O,v,(\prec_a)_{a\in A}\rangle$, let the subgame perfect preferences\footnote{Note that the translation of preferences in the following definition does not
preserve acyclicity. Preservation could be ensured, \textit{e.g.}, by
giving the nodes a linear "priority" order, in a lexicographic fashion.
This, however, would complicate the definition against little benefit for
the point that we want to make.} $\mathalpha{\prec}_a^{sgp} \subseteq C^{C^*} \times C^{C^*}$ be defined by $\sigma \prec_a^{sgp} \sigma'$ iff $\exists \gamma \in C^*$ such that $p^{\gamma}(\sigma) \prec_a p^{\gamma}(\sigma)$.

The subgame perfect equilibria of $\langle A,C,d,O,v,(\prec_a)_{a\in A}\rangle$ are the Nash equilibria \linebreak of $\langle A, (C^{d^{-1}(\{a\})})_{a \in A}, (\prec_a^{sgp})_{a \in A}\rangle$.
\end{definition}

We consider a further variant, namely the \emph{infinite sequential games with real-valued payoffs}, which can (but do not have to) be understood as a special case of infinite sequential games.
\begin{definition}
\label{def:payout}
An infinite sequential game with real-valued payoffs is a tuple $\langle A, C, d, (f_a)_{a \in \mathbb{N}}\rangle$ where $A$, $C$, $d$ are as above, and $f_a : C^\omega \to \mathbb{R}$ is the payoff function of player $a$.

Such a game can be identified with the infinite sequential game $\langle A, C, d, \mathbb{R}^A, v,(\prec_a)_{a\in A}\rangle$ where $v(p) = (f_a(p))_{a \in A}$ and for $x, y \in \mathbb{R}^A$, we set $x \prec_a y$ iff $x_a < y_a$.
\end{definition}

As with the introduction of subgame perfect equilibria, we can consider infinite sequential games with real-valued payoffs as infinite sequential games in a different way, which then gives rise to another commonly studied equilibrium concept, namely $\varepsilon$-Nash equilibria. Depending on how we then translate from infinite sequential games to abstract games, we obtain also $\varepsilon$-subgame perfect equilibria. Given some $\varepsilon > 0$, we define the relation $\prec_a^\varepsilon \subseteq \mathbb{R}^A \times \mathbb{R}^A$ by $x \prec_a^\varepsilon y$ iff $y_a - x_a > \varepsilon$. Using $\prec_a^\varepsilon$ in place of $\prec_a$ in Definition \ref{def:payout} then provides the above-mentioned equilibrium notions.

For infinite sequential games with real-valued payoffs, every Nash equilibrium (w.r.t.~the standard preferences) is an $\varepsilon$-Nash equilibrium; and every subgame perfect equilibrium is an $\varepsilon$-subgame perfect equilibrium. For infinite sequential games, every subgame perfect equilibrium is a Nash equilibrium, in particular, any $\varepsilon$-subgame perfect equilibrium is an $\varepsilon$-Nash equilibrium.

We use \emph{antagonistic game} to refer to two-player games with preferences satisfying $\prec_a = \prec_b^{-1}$, where $x \prec^{-1} y \Leftrightarrow y \prec x$.

We proceed to recall a few more notions that are only tangentially related to the formulation of our results, but that do show up in the proofs.
\begin{definition}
\label{def:2pwinl}
A two-player win-loose game is a tuple $\langle C, D, W\rangle$ with $D \subseteq C^*$ and $W \subseteq C^\omega$. It corresponds to the infinite sequential game $\langle \{a,b\}, C, d, \{0,1\}, v, \{<, <^{-1}\}\rangle$ where $d$ is defined via $d^{-1}(\{a\}) = D$ and $v$ is defined via $v^{-1}(\{1\}) = W$.
\end{definition}

Finally, we will extend the notion of the induced play. Given some partial function $s : \subseteq C^* \to C$, we define the consistency set $P(s) \subseteq C^\omega$ by: \[P(s) = \{p(\sigma) \mid \sigma : C^* \to C \wedge \sigma|_{\dom(s)} = s\}\]

\subsection*{Pareto-optimality}

Pareto-optimality provides a notion of social desirability in game theory, and can be used both to pick particularly \emph{nice} equilibria, and to investigate whether the strategic interaction is costly in some sense\footnote{Similar to the (quantitative) \emph{price of stability}, see \cite{roughgarden}.}.  The fundamental idea is that a Pareto-optimal outcome cannot be improved to everyone's satisfaction. Pareto-optimality is often only defined for linear preferences (or, slightly more general, strict weak orders), and its extension to general preferences is not obvious. The two natural choices are:

\begin{definition}
An outcome is \emph{realizable} in some game, if it is assigned to some sequence of choices. We call an outcome $o$ \emph{Pareto-optimal}, if there is no other realizable outcome $q$ such that for some player $a$ we find $o \prec_a q$ and for no player $b$ we have $q \prec_a o$. We call an outcome $o$ \emph{weakly Pareto-optimal}, if there is no other realizable outcome $q$ such that for some player $a$ we find $o \prec_a q$ and for all players $b$ we have that $o \preceq_b q$.
\end{definition}

Note that for linear preferences, both notions coincide. We shall call a Nash equilibrium (weakly) Pareto-optimal iff it induces a (weakly) Pareto-optimal outcome.

\section{The continuity argument}
\label{sec:continuityargument}
A strong transfer result can be obtained using topological arguments alone, with the reasoning being particularly well-adapted to a formulation in synthetic topology (originally \cite{escardo}, \cite{pauly-synthetic-arxiv} for a short introduction). Consider games in normal form, with potentially countably many agents with strategy spaces $\mathbf{S}_1, \mathbf{S}_2, \ldots$. Our first condition is that each $\mathbf{S}_i$ be compact (subsequently, by Tychonoff's Theorem, also $\Pi_{i \in \mathbb{N}} \mathbf{S}_i$). This restriction is very common and usually combined with continuity of the outcome function, as it avoids pathological games such as
\emph{pick-the-largest-natural-number}. Our second condition is that each preference relation $\prec_i$ is open (as a subset of $\left ( \Pi_{i \in \mathbb{N}} \mathbf{S}_i \right ) \times \left ( \Pi_{i \in \mathbb{N}} \mathbf{S}_i \right )$).\footnote{There actually is a third condition, that any strategy space is overt. A space $\mathbf{X}$ is overt, if $\{\emptyset\} \subseteq \mathcal{O}(\mathbf{X})$ is a closed set, i.e.~if there is a way to detect non-emptiness of open subsets. This condition could only ever fail in a constructive reading, but is always valid for topological spaces in classical logic. Synthetic topology however would also allow us to read \emph{continuous map} to mean \emph{computable map}, in which case overtness becomes non-trivial. In this reading, though, we actually obtain an algorithmic result.}  In the reading of synthetic topology, this means that any agent will be able to eventually confirm that he prefers a given strategy profile to another, provided he does indeed do so. We shall call a class of games $\mathcal{G}$ satisfying these conditions (in a uniform way) to be compact-strategies, open-preferences. Uniformity here means that we assume a topology on $\mathcal{G}$ such that the function mapping a game to the preferences is continuous itself.

We will write $\mathcal{O}(\mathbf{X})$ for the hyperspace of open subsets of $\mathbf{X}$, and $\mathcal{K}(\mathbf{X})$ for the hyperspace of compact sets. By $\mathcal{C}(\mathbf{X}, \mathbf{Y})$ we denote the space of continuous functions from $\mathbf{X}$ to $\mathbf{Y}$, in particular $\mathcal{C}(\mathbb{N},\mathbf{X})$ we denote the space of sequences in $\mathbf{X}$. For precise definitions, see \cite{pauly-synthetic-arxiv}. There we also find that the following operations are continuous:
\begin{enumerate}
\item \footnote{In general, the continuity of this map would require $\mathbf{Y}$ to be overt. As explained above, in a classical reading, this condition is always satisfied.} $\exists : \mathcal{O}(\mathbf{X} \times \mathbf{Y}) \to \mathcal{O}(\mathbf{X})$, defined \\ by $\exists(U) = \{x \in \mathbf{X} \mid \exists y \in \mathbf{Y} \ (x,y) \in U\}$
\item $\bigcup : \mathcal{C}(\mathbb{N}, \mathcal{O}(\mathbf{X})) \to \mathcal{O}(\mathbf{X})$
\item $^C : \mathcal{O}(\mathbf{X}) \to \mathcal{K}(\mathbf{X})$, provided that $\mathbf{X}$ is compact.
\item $\textrm{NonEmptyValue} : \mathcal{C}(\mathbf{X}, \mathcal{K}(\mathbf{Y})) \to \mathcal{O}(\mathbf{X})$ defined by\\$\textrm{NonEmptyValue}(f) = \{x \in \mathbf{X} \mid f(x) \neq \emptyset\}$.
\end{enumerate}

\begin{theorem}
\label{theo:cont}
Let $\mathcal{G}$ be compact-strategies, open-preferences, and let $\mathcal{G}' \subseteq \mathcal{G}$ be a dense subclass. If every $G \in \mathcal{G}'$ has a Nash equilibrium, then every $G \in \mathcal{G}$ has a Nash equilibrium.
\begin{proof}
By combining our continuous operations, we may obtain the set of all games in $\mathcal{G}$ with a Nash equilibrium as an open set in the following way:
\[\mathcal{NE} := \textrm{NonEmptyValue}\left ( G \mapsto \left (\bigcup_{i \in \mathbb{N}} \exists (\prec_i^G) \right )^C\right ) \in \mathcal{O}(\mathcal{G})\]
Formulating the individual steps in words: with $\prec_i$ being open\footnote{And $\mathbf{S}_i$ being overt, see above.}, we immediately obtain that the set of all strategy profiles such that player $i$ has a better response is uniformly open in the game. Taking the union over all players again yields an open set, which complement now is the closed set $\textrm{NE}(G)$ of all Nash equilibria of the respective game $G$. As this is a subset of the compact space $\left ( \Pi_{i \in \mathbb{N}} \mathbf{S}_i \right )$, we can even treat $\textrm{NE}(G)$ as a compact set uniformly in $G$. By the synthetic definition of compactness, we obtain that $\{G \in \mathcal{G} \mid \textrm{NE}(G) = \emptyset\} \subseteq \mathcal{G}$ is an open set.

Because we have assumed $\mathcal{G}'$ to be dense in $\mathcal{G}$, we see that if any game in $\mathcal{G}$ would fail to have a Nash equilibrium, this would imply that some game in $\mathcal{G}'$ would fail, too, contrary to the assumption.
\end{proof}
\end{theorem}

\begin{lemma}
\label{lemma:contsubgame}
Consider sequential games with continuous payoff-functions and finite choices sets. We find:
\begin{enumerate}
\item The subgame-perfect preferences produce a compact-strategies, open-preference class $\mathcal{S}$.
\item The games with payoffs fully determined after finitely many moves are a dense subset $\mathcal{S}^f \subseteq \mathcal{S}$.
\item All games in $\mathcal{S}^f$ have a Nash equilibrium.
\end{enumerate}
\begin{proof}
\begin{enumerate}
\item First, note that the mapping $(\gamma, \sigma) \mapsto p_\gamma(\sigma)$ is continuous, so for any fixed $\gamma \in C^*$,\\ $\{(\sigma, \sigma') \mid f_a(p_\gamma(\sigma)) <  f_a(p_\gamma(\sigma'))\}$ is an open set. By taking countable union, we learn that $\prec_a^{sgp}$ is open. Compactness and overtness of the strategy spaces are straightforward.
\item As the argument in $(1)$ is uniform in the continuous functions $f_a$, it suffices to argue that the payoff functions $f : C^\omega \to \mathbb{R}$ depending only on some finite prefix of the input are dense in $\mathcal{C}(C^\omega, \mathbb{R})$. A countable base for the applicable topology is found in all $$\{f \mid \forall \gamma \in C^k, p \in C^\omega \ . \ p_{\leq k} = \gamma \Rightarrow f(p) \in (x_\gamma, y_\gamma)\}$$ for $$k \in \mathbb{N}, x_{(\cdot)}, y_{(\cdot)} : C^k \to \mathbb{Q}$$
    A base element is non-empty, iff $\forall \gamma \in C^k \  x_{\gamma} < y_{\gamma}$; and then it will contain the function $f_0 : C^\omega \to \mathbb{R}$ defined via $f_0(p) = \frac{1}{2}\left(x_{p_{\leq k}} + y_{p_{\leq k}}\right )$, which clearly depends only on the prefix of length $k$ of its argument.
\item As the actions of the players beyond the finite prefix determining the outputs is irrelevant, and taking into consideration the definition of the subgame-perfect preferences, the claim is that any finite game in extensive form has a subgame perfect equilibrium. This well-known result (\name{Kuhn} \cite{kuhn53}) is easily proven by backwards induction: let the players who move last pick an optimal (for them) choice. Then the players who move second-but-last have guaranteed outcomes associated with their moves, so they can optimize, and so on.
\end{enumerate}
\end{proof}
\end{lemma}

\begin{corollary}[(\footnote{This extends \cite[Corollary 4.2]{fudenberg} from finitely many players to countably many players.})]
\label{corr:contseque}
Any sequential game with continuous payoff functions and finitely many choices has a subgame perfect Nash equilibrium.
\end{corollary}

We shall consider a number of variations/extensions. First, we investigate stochastic\footnote{In this, we are answering a question raised by \name{Assal\'e Adj\'e} at CSL-LiCS.} infinite sequential games with continuous payoff functions as a variation of Definition \ref{def:payout}: a stochastic infinite sequential game with real-valued payoffs is a tuple $\langle A, C, d, (f_a)_{a \in \mathbb{N}}, P\rangle$ where $A$, $C$, $f_a$ are as above, $\mathbf{n} \notin A$, $d : C^* \to A \cup \{\mathbf{n}\}$ and $P$ assigns a probability distribution over $C$ to each $w \in d^{-1}(\{\mathbf{n}\})$.

The notion of the play induced by a strategy profile is replaced by a probability distribution over plays induced by a strategy profile: essentially, we choose according to the strategy profile in vertices controlled by a player, and stochastically according to $P(w)$ in vertices controlled by $\mathbf{n}$ (nature). A player prefers a strategy profile to another one if the expected value of his payoff function regarding the former induced probability distribution exceeds the expected value regarding the latter. We obtain a notion of a subgame perfect equilibrium as before.

\begin{lemma}
\label{lemma:contstoch}
Consider stochastic infinite sequential games with continuous payoff-functions and finite choices sets. We find:
\begin{enumerate}
\item The subgame-perfect preferences produce a compact-strategies, open-preference class $\mathcal{SS}$.
\item The games with payoffs fully determined after finitely many moves are a dense subset $\mathcal{SS}^f \subseteq \mathcal{SS}$.
\item All games in $\mathcal{SS}^f$ have a Nash equilibrium.
\end{enumerate}
\begin{proof}
\begin{enumerate}
\item For a synthetic approach to continuity on probability measures etc, see \cite{collins4}. In particular, the map from strategy profiles to induced probability distributions is easily seen to be continuous; and integration is as well. The remaining argument proceeds as in Lemma \ref{lemma:contsubgame} (1).
\item The stochastic vertices do not impact the argument, so it works as in Lemma \ref{lemma:contsubgame} (2).
\item As in Lemma \ref{lemma:contsubgame} (3), we use backwards inductions: the value of a given leaf for a player is immediately obtained from the payoff function. In any vertex controlled by a player, he picks a choice guaranteeing him the optimal value, and the value of that vertex for any player $p$ is identical to the value of the corresponding child for $p$. The value of a nature vertex is the expected value of its children according to the probability distribution given by $P$. This is easily seen to yield a subgame perfect equilibrium.
\end{enumerate}
\end{proof}
\end{lemma}

\begin{corollary}
Stochastic infinite sequential games with continuous payoff-functions and finite choices sets have subgame-perfect equilibria.
\end{corollary}

In a similar fashion we can consider multi-player multi-outcome Blackwell games \cite{blackwell,martin2} with continuous payoff functions (into $\uint$). In a Blackwell game, in each round a finite subset $A_v$ of $m$ players act jointly and probabilistic and thus determine the next vertex in the tree. In particular, we are now dealing with a tree with branching $C^{m}$ instead of just $C$. The map $d$ indicating who plays now goes into the $m$-element subsets of $A$ rather than identifying just a single player. Strategies are functions $s_a :\subseteq (C^{m})^* \to P(C)$, where $P(C)$ shall denote the set of probability distributions over $C$, strategy profiles are of the form $s :\subseteq A \times (C^m)^* \to P(C^m)$ and payoff functions go $f_a : (C^m)^\omega \to \uint$. Each strategy profile induces a probability distribution over $(C^m)^\omega$ (also from a given vertex onwards), payoffs for players are then again the expected values of their payoff functions given the induced probability distribution.

\begin{lemma}
\label{lemma:blackwell}
Consider Blackwell games with continuous payoff-functions, countably many agents and a finite choice set. We find:
\begin{enumerate}
\item The subgame-perfect preferences produce a compact-strategies, open-preference class $\mathcal{B}$.
\item The games with payoffs fully determined after finitely many moves (and thus depending only on the choices made by finitely many players) are a dense subset $\mathcal{B}^f \subseteq \mathcal{B}$.
\item All games in $\mathcal{B}^f$ have a Nash equilibrium.
\end{enumerate}
\begin{proof}
\begin{enumerate}
\item Very much like Lemma \ref{lemma:contstoch} (1). The induced probability distribution depends continuously on the strategy profile, and the expected value is a continuous map. It is unproblematic that $C$ is replaced by $\mathcal{P}(C)$ now, as this space is still compact.
\item As in Lemma~\ref{lemma:contsubgame} (2), we only need to argue that the continuous payoff functions determined by their values on $(C^{m})^n$ for some $n \in \mathbb{N}$ are dense in $\mathcal{C}((C^{m})^\omega, \uint)$.  This in turn follows just as before, taking into account the definition of the product topology.
\item We can ignore the players whose actions do not impact the payoff functions, and the moves in the game taking place after the payoffs have been determined. Thus we are dealing with a finite tree, and can use backwards induction again. The leaves of the finite tree are assigned payoff values from the payoff functions. The preceding layer of vertices now corresponds to usual normal form games (with finitely many players and finitely many choices), who have Nash equilibria by \name{Nash}' classic result \cite{nash}. Pick a Nash equilibrium for each such vertex, and assign it the corresponding outcomes as values. Then the layer preceding this one can be handled in the same way, etc. -- so the claim follows by backwards induction.
\end{enumerate}
\end{proof}
\end{lemma}

\begin{corollary}[(\footnote{This extends \cite[Theorem 6.1]{fudenberg} from finitely many players to countably many players.})]
Any multi-outcome Blackwell game with countably many players and continuous payoff functions has a subgame-perfect Nash equilibrium.
\end{corollary}

The arguments for Lemma \ref{lemma:blackwell} remain unaffected if the number of agents playing at a particular vertex varies with the vertex (this just increases notational complexity). We can even let all players act at each time, as long as we then restrict their strategies to continuously depend on the choices (note that this would produce a tree with branching factor $2^{\aleph_0}$). We can also turn one player into a designated nature player, which plays according to some prespecified strategy rather than according to preferences, and thus obtain:

\begin{corollary}
Any multi-outcome stochastic Blackwell game with countably many players and continuous payoff functions has a subgame-perfect Nash equilibrium.
\end{corollary}

Concurrent games (in the sense of \name{Winskel}) may be another interesting class of games to apply this approach to. Determinacy questions of these games were discussed in \cite{winskel,winskel2}.

\section{Infinite sequential games with infinitely ascending preferences}
\label{sec:ascendingpreferences}
As soon as the continuity requirement for the payoff function (or, more generally, the openness of the preferences) is dropped, Nash equilibria may fail to exist. We provide a generic folklore counterexample, and will proceed to demonstrate that the underlying feature is essential for the failure of existence of Nash equilibria. The counterexample only requires a single player, and its payoff function is in a sense \emph{the least discontinuous} payoff function, and in particular is $\Delta_2^0$-measurable.
\begin{example}
\label{example:nonash}
Let the payoff function $P : \Cantor \to \uint$ for the single player be defined by $P(1^n0p) = \frac{n}{n+1}$ for all $p\in\Cantor$ and $P(1^\mathbb{N}) = 0$. As $P$ does not attain its supremum, the resulting game cannot have a Nash equilibrium.

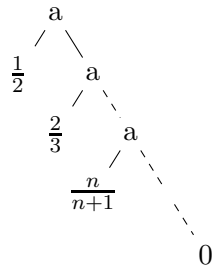
\center
\begin{tikzpicture}[level distance=8mm] %%% recursive definition of a tree!
\node{a}[sibling distance=10mm]
	child{node{$\frac{1}{2}$}[sibling distance=10mm]}
    child{node{a}[sibling distance=10mm]
		child{node{$\frac{2}{3}$}}
		child{node{a} [sibling distance=10mm] edge from parent[dashed]
			child{node{$\frac{n}{n+1}$} edge from parent[solid]} %%% "solid" is supposed to be the default value, but it seems it has to be specified after using "dashed".
			child{node{} edge from parent[dashed]
				child{node{}edge from parent[draw=none]}
				child{node{$0$} edge from parent[dashed]}
			}
		}	
	};
\end{tikzpicture}
\end{example}

We proceed to show in particular that the presence of a converging sequence of plays $(p^n)_{n \in \mathbb{N}}$ such that a player prefers $p^{n+1}$ to $p^n$, but prefers any $p^n$ to $\lim_{i \to \infty} p^i$, is a crucial feature of the example above to have no Nash equilibrium. The proof will be an adaption of the main result of \cite{leroux3} by the first author. Under the additional assumption of antagonistic preferences in a two-player game, we can even obtain subgame perfect equilibria.

In this section the preferences of the players are restricted to strict weak orders, so we recall their definition below.
\begin{definition}[Strict weak order]
 \label{def:strictweakorder}
A relation $\prec$ is called a \emph{strict weak order} if it satisfies:
\[\begin{array}{l@{\quad}l}
\forall x,\quad \neg(x\prec x)  \\
\forall x,y,z, \quad x\prec y \,\wedge\, y\prec z\,\Rightarrow\, x\prec z  \\
\forall x,y,z, \quad \neg(x\prec y) \,\wedge\, \neg(y\prec z)\,\Rightarrow\, \neg(x\prec z)
\end{array}\]
\end{definition}

Definition~\ref{defn:ag-IAC} below slightly rephrases Definitions 2.3 and 2.5 from~\cite{leroux3}: the guarantee of a player is the smallest set of outcomes that is upper-closed w.r.t. the strict-weak-order preference of the player and includes every incomparability classe (of the preference) that contains any outcome compatible with a given strategy of the player in the subgame at a given node of a given infinite sequential game. The best guarantee of a player consists of the intersection of all her guarantees over the set of strategies.

\begin{definition}[Agent (best) guarantee]\label{defn:ag-IAC}
 Let $\langle A,C,d,O,v,(\prec_a)_{a\in A}\rangle$ be a game where the $\prec_a$ are strict weak orders.
\[\begin{array}{l}
\forall a\in A,\forall \gamma\in C^*,\forall s : d^{-1}(a) \to C,\quad g_a(\gamma,s):=\\
 \{o\in O\,\mid\,\exists p\in P(s|_{\gamma C^\omega}) \cap \gamma C^\omega,\, \neg (o \prec_a v(p))\}\\
G_a(\gamma):=\bigcap_{s}g_a(\gamma,s)
\end{array}\]
We write $g_a(s)$ and $G_a$ instead of $g_a(\gamma,s)$ and $G_a(\gamma)$ when $\gamma$ is the empty word.
\end{definition}

Lemma 2.4. from~\cite{leroux3} still holds without major changes in the proofs, so we do not display it, but note that when speaking about $\prec_a$-terminal intervals (which are upper-closed sets), we now actually refer to the terminal intervals of the lift of $\prec_a$ from outcomes to the equivalence classes of outcomes induced by the strict weak order. Also, we collect some more useful facts in Observation~\ref{obs:G} below.

\begin{observation}\label{obs:G}
Let $\langle A,C,O,d,v,(\prec_a)_{a\in A}\rangle$, let $a\in A$, assume that $\prec_a$ is a strict weak order, and let $\gamma\in C^*$.
\begin{enumerate}
\item $d(\gamma) \neq a \,\Rightarrow\, G_a(\gamma) = \cup_{c\in C} G_a(\gamma\cdot c)$
\item $d(\gamma) = a \,\Rightarrow\, G_a(\gamma) = \cap_{c\in C} G_a(\gamma\cdot c)$
\item $d(\gamma) = a \,\wedge\, |C|<\infty \,\Rightarrow\,\exists c\in C,\, G_a(\gamma) = G_a(\gamma\cdot c)$
\end{enumerate}

\begin{proof}
For example, for $2.$ note that $G_a(\gamma) = \cap_{s}g_a(\gamma,s) = \cap_{c\in C} \cap_{s(\gamma)=c} g_a(\gamma,s)=\cap_{c\in C}\cap_{s}g_a(\gamma\cdot c,s) \\ =\cap_{c\in C}G_a(\gamma\cdot c)$.
\end{proof}
\end{observation}

This section's proofs of existence of equilibrium rely on each player having a (minimax-style) optimal strategy if all other players team up against her. Lemma~\ref{lem:max-guarant0} below provides a sufficient condition for such strategies to exist, \textit{i.e.} for the best guarantee to be witnessed.

\begin{lemma}\label{lem:max-guarant0}
Let $\langle A,C,O,d,v,(\prec_a)_{a\in A}\rangle$ be a game where $C$ is
finite, let $a\in A$, and let us assume the following.
\begin{enumerate}
\item $\prec_a$ is a strict weak order.
\item\label{cond:max-guarant02} For every play $p\in C^\omega$, increasing $\varphi: \mathbb{N} \to \mathbb{N}$, and sequence $(s_n)_{n\in\mathbb{N}}$ of strategies for Player $a$, if $g_a(p_{<\varphi(n+1)},s_{n+1}) \subsetneq g_a(p_{<\varphi(n)},s_n)$ for all $n\in\mathbb{N}$, then $v(p)\in \cap_{n\in\mathbb{N}}g_a(p_{<\varphi(n)},s_n)$.
\end{enumerate}
Then for all $\gamma\in C^*$ there exists $s\in S_a$ such that $g_a(\gamma,s)=G_a(\gamma)$.

\begin{proof}
Wlog we only prove that there exists $s\in S_a$ such that $g_a(s)=G_a$, \textit{i.e.} where the $\gamma$ from the claim is the empty word. Let $s_0:\, d^{-1}(a) \to C$ be a strategy for Player $a$ and let us build inductively a sequence $(s_n)_{n\in\mathbb{N}}$ of strategies for Player $a$, as follows, where case $3.$ implicitly invokes Observation~\ref{obs:G}.
\begin{itemize}
\item Let $s_{n+1}|_{C^{<n}} := s_{n}|_{C^{<n}}$.
\item For all $\gamma\in C^{n}\backslash d^{-1}(a)$, let $s_{n+1}|_{\gamma
C^*} := s_{n}|_{\gamma C^*}$.
\item For all $\gamma\in C^{n} \cap d^{-1}(a)$,
\begin{enumerate}
\item if $g_a(\gamma,s_n) \subseteq G_a$ then let $s_{n+1}|_{\gamma C^*}
: = s_{n}|_{\gamma C^*}$,
\item if $G_a \subsetneq g_a(\gamma,s_n)$ and there exists
$\mu:\,d^{-1}(a) \cap \gamma C^* \to C$ such that $g_a(\gamma,\mu)
\subseteq G_a$, let $s_{n+1}|_{\gamma C^*} := \mu$,
\item otherwise\footnote{Note that due to the properties of a strict weak order, the sets of the form $g_a(\gamma,s)$ and $G_a(\gamma)$ are linearly ordered by inclusion $\subseteq$. Thus, $G_a \subsetneq g_a(\gamma,s_n)$ holds in this case, too.} let $s_{n+1}(\gamma) := c$ such that
$G_a(\gamma\cdot c)=G_a(\gamma)$, and let $s_{n+1}|_{\gamma C C^*} :=
s_{n}|_{\gamma C C^*}$.
\end{enumerate}
\end{itemize}
Let $s$ be the limit strategy of the sequence $(s_n)_{n\in\mathbb{N}}$ and first note that, using Observation~\ref{obs:G}, one can prove by induction on $\gamma$ that $G_a(\gamma) \subseteq G_a$ for every
$\gamma\in C^*$ that is compatible with $s$. Next, let $p\in P(s)$ be a path
compatible with $s$. If $p$ has a prefix $\gamma$ that fell into Cases
$1.$ or $2.$ during the recursive construction above, then $v(p)\in G_a$, so let us now assume that case $3.$ applies at every node $p_{<n}\in d^{-1}(a)$. If such nodes are finitely many, let $p_{<n}$ be the deepest one, so $d(p_{n+1+k}) \neq a$ for all $k\in\mathbb{N}$, and $v(p)\in g_a(p_{n+1},t)$ for all strategies $t$ for $a$, so $v(p) \in G_a(p_{<n+1}) = \bigcap_{t}g_a(p_{<n+1},t)$. So $v(p)\in G_a$ since $G_a=G_a(p_{<n+1})$ by Case $3$. Let us now assume that such nodes are infinitely many. If $G_a(p_{<n}) \subsetneq G_a$ for some $p_{<n}\in d^{-1}(a)$, there exists $\mu: d^{-1}(a) \cap p_{<n}C^* \to C$ such that $G_a(p_{<n}) \subseteq g_a(p_{<n},\mu)
\subsetneq G_a$ since $G_a(p_{<n}) = \bigcap_{t}g_a(p_{<n},t)$ by
definition, which would mean that Case $1.$ or $2.$ applies; so
$G_a(p_{<n})=G_a$ for all $p_{<n}\in d^{-1}(a)$, and subsequently for all $n$. Also, the best guarantee is never
witnessed (through Case $2.$) at any node $p_{<n}\in d^{-1}(a)$,  and subsequently for all $n$. If $v(p)\notin G_a$, the previous two remarks allow us to build inductively a sequence $(t_n)_{n\in\mathbb{N}}$ of strategies for $a$ such that $v(p) \notin g_a(p_{<0},t_0)$ and $g_a(p_{<n+1},t_{n+1}) \subsetneq g_a(p_{<n},t_n)$ for all $n\in\mathbb{N}$, which would imply $v(p)\in \cap_{n\in\mathbb{N}}g_a(p_{<n},t_n)$ by assumption of the lemma, contradiction.
\end{proof}
\end{lemma}

More specifically, this section's proofs of existence of equilibrium rely on each player having a strategy that is (minimax-style) optimal for every subgame at once. Lemma~\ref{lem:max-guarant1} below gives a sufficient condition for such strategies to exists. It relies on Condition~\ref{cond:max-guarant13}, which is also the conclusion in Lemma~\ref{lem:max-guarant0}, to perform a key quantifier inversion. Note that Condition~\ref{cond:max-guarant12} in Lemma~\ref{lem:max-guarant1} is weaker than that in Lemma~\ref{lem:max-guarant0}, and that finiteness of $C$ is used in Lemma~\ref{lem:max-guarant0} only. It is the main reason why Lemmas~\ref{lem:max-guarant0} and \ref{lem:max-guarant1} are not merged.

\begin{lemma}\label{lem:max-guarant1}
Let $\langle A,C,O,d,v,(\prec_a)_{a\in A}\rangle$ be a game, let $a\in A$,
and let us assume the following.
\begin{enumerate}
\item\label{cond:max-guarant11} $\prec_a$ is a strict weak order.
\item\label{cond:max-guarant12} For every play $p\in C^\omega$ and increasing $\varphi: \mathbb{N} \to \mathbb{N}$, if $d(p_{<\varphi(n)}) = a$ and $G_a(p_{<\varphi(n+1)}) \subsetneq G_a(p_{<\varphi(n)})$ for all $n\in\mathbb{N}$, then $v(p)\in \cap_{n\in\mathbb{N}}G_a(p_{<\varphi(n)})$.
\item\label{cond:max-guarant13} For all $\gamma\in C^*$ there exists $s\in S_a$ such that
$g_a(\gamma,s)=G_a(\gamma)$.
\end{enumerate}
Then there exists $s$ such that $g_a(\gamma,s)=G_a(\gamma)$ for all
$\gamma\in C^*$.

\begin{proof}
We proceed similarly as in the proof of Lemma~\ref{lem:max-guarant0}. Let
$s_0$ be a strategy for Player $a$ and let us build inductively a sequence
$(s_n)_{n\in\mathbb{N}}$ of strategies for Player $a$. The recursive
definition below is different from the one in the proof of
Lemma~\ref{lem:max-guarant0} in three respects: the three occurrences of
$G_a$ in Cases $1.$ and $2.$ are replaced with $G_a(\gamma)$. Case $3.$ is
deleted since it never applies by assumption. Finally, two inclusions are
replaced with equalities.
\begin{itemize}
\item Let $s_{n+1}|_{C^{<n}} := s_{n}|_{C^{<n}}$
\item For all $\gamma\in C^{n}\backslash d^{-1}(a)$, let $s_{n+1}|_{\gamma
C^*} := s_{n}|_{\gamma C^*}$.
\item For all $\gamma\in C^{n} \cap d^{-1}(a)$,
\begin{enumerate}
\item if $g_a(\gamma,s_n) = G_a(\gamma)$ then let $s_{n+1}|_{\gamma C^*}
: = s_{n}|_{\gamma C^*}$,
\item if $G_a(\gamma) \subsetneq g_a(\gamma,s_n)$, let $s_{n+1}|_{\gamma
C^*} := \mu$ where $\mu:\,d^{-1}(a) \cap \gamma C^* \to C$ is such that
$g_a(\gamma,\mu) = G_a(\gamma)$.
\end{enumerate}
\end{itemize}
Let $s$ be the limit strategy of the sequence $(s_n)_{n\in\mathbb{N}}$ and first note that, using Observation~\ref{obs:G}, one can prove by induction on $\gamma$ that $G_a(\gamma) \subseteq G_a$ for every
$\gamma\in C^*$ that is compatible with $s$. Next, let $p\in P(s)$ be a path
compatible with $s$. Due to the uniformity of the recursive definition,
it suffices to show that $v(p)\in G_a$ to prove the full statement.

If Case $2.$ applies only finitely many times
in the construction of $s$, the sequence $(s_n|_{\{\gamma\in C^*\,\mid\,p\in \gamma C^\omega\}})_{n\in\mathbb{N}}$ is eventually
constant, so $v(p)\in g_a(p_{<n},s_n) = G_a(p_{<n}) \subseteq G_a$ for some
$n$. Otherwise, there exists an increasing function $\varphi:\,\mathbb{N}
\to \mathbb{N}$  with $d(p_{<\varphi(n)}) = a$ and $G_a(p_{<\varphi(n+1)}) \subsetneq
G_a(p_{<\varphi(n)})$ for all $n\in \mathbb{N}$, so $v(p)\in  \cap_{n\in\mathbb{N}}G_a(p_{<\varphi(n)}) \subseteq  G_a(p_{<\varphi(0)}) \subseteq G_a$.
\end{proof}
\end{lemma}

Theorems~\ref{thm:zs-fc-spe} and \ref{thm:zs-fo-spe} below both prove the existence of subgame perfect equilibria for antagonistic games, either when the choice set $C$ is finite, or when the outcome set $O$ is finite. Since their proofs are similar, most is factored out in Lemma~\ref{lem:abstract-zero-sum-SPE} below.

\begin{lemma}\label{lem:abstract-zero-sum-SPE}
Let $\langle \{a,b\},C,O,d,v,\{\prec,\prec^{-1}\}\rangle$ be a two-player game. Let $\Gamma\subseteq\mathcal{P}(C^\omega)$ and assume the following.
\begin{enumerate}
\item $\prec$ is a strict weak order.
\item\label{cond:azs-spe1} For every play $p\in C^\omega$ and increasing sequence $\varphi: \mathbb{N} \to \mathbb{N}$, if $d(p_{<\varphi(n)}) = a$ and $G_a(p_{<\varphi(n+1)}) \subsetneq G_a(p_{<\varphi(n)})$ for all $n\in\mathbb{N}$, then $v(p)\in \cap_{n\in\mathbb{N}}G_a(p_{<\varphi(n)})$.
\item\label{cond:azs-spe2} For all $\gamma\in C^\omega$, there exists $s$ such that $g_a(\gamma,s)=G_a(\gamma)$ (resp. $g_b(\gamma,s)=G_b(\gamma)$).
\item\label{cond:azs-spe3} For all non-empty closed $E\subseteq C^\omega$, there are $\prec$-extremal elements in $v[E]$.
\item\label{cond:azs-spe4} For every $\prec$-extremal interval $I$ and $\gamma\in C^*$, we have $\left (v^{-1}[I]\cap \gamma C^\omega \right ) \in\Gamma$.
\item\label{cond:azs-spe5} The game $\langle C,D,W\rangle$ is determined for all $W\in\Gamma$, $D \subseteq C^*$.
\end{enumerate}
Then the game $\langle \{a,b\},C,O,d,v,\{\prec,\prec^{-1}\}\rangle$ has a subgame perfect equilibrium.

\begin{proof}
By invoking Lemma~\ref{lem:max-guarant1} once for Player $a$ and once for Player $b$, let us build a strategy profile $s:C^*\to C$, such that $g_X(\gamma,s_X)= G_X(\gamma)$ for all $\gamma\in C^*$ and $X\in\{a,b\}$. Let $\gamma\in C^*$ and let us prove that $G_a(\gamma)\cap G_b(\gamma)=\{\mathrm{min}_<(G_a(\gamma))\}=\{\mathrm{max}_<(G_b(\gamma))\}$. Consider the game $\langle C,D,W\rangle$ (as in Definition \ref{def:2pwinl}) where the winning set is defined by $W:=\{\alpha\in \gamma C^{\omega}\,\mid\,v(\alpha)\in G_a(\gamma)\backslash\{\mathrm{min}_<(G_a(\gamma))\}\}$ and where Player $a$ owns exactly the nodes in $D:= (C^* \setminus \gamma C^*) \cup (d^{-1}(\{a\}) \cap \gamma C^*)$. By Assumption~\ref{cond:azs-spe4} the set $W$ is in $\Gamma$, so by Assumption~\ref{cond:azs-spe5} the game $\langle C,D,W\rangle$ is determined. By definition of the best guarantee, Player $a$ has no winning strategy for this game, so Player $b$ has a winning strategy, which means that $G_b(\gamma)\subseteq \{\mathrm{min}_<(G_a(\gamma))\}\cup O\backslash G_a(\gamma)$. Since $G_a(\gamma)\cap G_b(\gamma)$ must be non-empty, otherwise the two guarantees are contradictory, $G_a(\gamma)\cap G_b(\gamma)=\{\mathrm{min}_<(G_a(\gamma))\}$. This means that the subprofile of $s$ rooted at $\gamma$ induces the outcome $\mathrm{min}_<(G_a(\gamma))$ (which equals $\mathrm{max}_<(G_b(\gamma))$ by symmetry), and it is optimal for both players.
\end{proof}
\end{lemma}

\begin{theorem}\label{thm:zs-fc-spe}
Let $\langle \{a,b\},C,O,d,v,\{\prec,\prec^{-1}\}\rangle$ be a two-player antagonistic game, where $C$ is finite. Let $\Gamma\subseteq\mathcal{P}(C^\omega)$ and assume the following.
\begin{enumerate}
\item $\prec$ is a strict weak order.
\item\label{thm:zs-fc-spe2} For every $p\in C^\omega$, sequence $(s_n)_{n\in\mathbb{N}}$ of strategies for $X\in\{a,b\}$, and increasing $\varphi: \mathbb{N} \to \mathbb{N}$, if $d(p_{<\varphi(n)}) = a$ and $g_X(p_{<\varphi(n+1)},s_{n+1}) \subsetneq g_X(p_{<\varphi(n)},s_n)$ for all $n\in\mathbb{N}$, then $v(p)\in \cap_{n\in\mathbb{N}}g_X(p_{<\varphi(n)},s_n)$.
\item For every $\prec$-extremal interval $I$ and $\gamma\in C^*$, we have $\left (v^{-1}[I]\cap \gamma C^\omega \right ) \in\Gamma$.
\item The game $\langle C,D,W\rangle$ is determined for all $W\in\Gamma$, $D \subseteq C^*$.
\end{enumerate}
Then the game $\langle \{a,b\},C,O,d,v,\{\prec,\prec^{-1}\}\rangle$ has a subgame perfect equilibrium.

\begin{proof}
By application of Lemma~\ref{lem:abstract-zero-sum-SPE}. (Note that Condition~\ref{thm:zs-fc-spe2} of Theorem~\ref{thm:zs-fc-spe} is an "upper bound" of Condition~\ref{cond:max-guarant02} of Lemma~\ref{lem:max-guarant0} and Condition~\ref{cond:max-guarant12} of Lemma~\ref{lem:max-guarant1}.) Condition~\ref{cond:azs-spe2} is proved by Lemma~\ref{lem:max-guarant0}. For Condition~\ref{cond:azs-spe3}, let $E$ be a non-empty closed subset of $C^\omega$, and let $T$ be the tree such that $[T] = E$. Consider the game where Player $a$ plays alone on $T$. Since Player $a$ can maximise her best guarantee by Lemma~\ref{lem:max-guarant0}, and since all her guarantees are singletons, $v[E]$ has a $\prec$-maximum. Likewise, it has a $\prec$-minimum, by considering Player $b$.
\end{proof}
\end{theorem}

\begin{theorem}\label{thm:zs-fo-spe}
Let $\langle \{a,b\},C,d,O,v,\{<,<^{-1}\}\rangle$ be an infinite sequential game where $O$ is finite and $<$ is a strict linear order. Let $\Gamma\subseteq\mathcal{P}(C^\omega)$ and assume the following.
\begin{enumerate}
\item\label{hyp:wo-ne-stable} $\forall O'\subseteq O,\forall\gamma\in C^*,\quad \{\alpha\in C^{\omega}\,\mid\,v(\gamma\alpha)\in O'\}\in\Gamma$
\item\label{hyp:det} The game $\langle C,D,W\rangle$ is determined for all $W\in\Gamma$ and $D\subseteq C^*$.
\end{enumerate}
Then the game $\langle \{a,b\},C,d,O,v,\{<,<^{-1}\}\rangle$ has a subgame perfect equililbrium.

\begin{proof}
by Lemma~\ref{lem:abstract-zero-sum-SPE} where Conditions~\ref{cond:azs-spe1}, \ref{cond:azs-spe2}, and \ref{cond:azs-spe3} hold by finiteness of $O$.
\end{proof}
\end{theorem}

\begin{corollary}\label{cor:qbd-spe}
Let $\langle \{a,b\},C,d,O,v,\{<,<^{-1}\}\rangle$ be an infinite sequential game where $O$ is finite and $<$ is a strict linear order. If $v^{-1}(o)$ is quasi-Borel for all $o\in O$, the game has a subgame perfect equilibrium.
\begin{proof}
From Theorem~\ref{thm:zs-fo-spe}, quasi-Borel determinacy~\cite{martin2}, and Lemma 3.1. in~\cite{leroux3}.
\end{proof}
\end{corollary}

In the remainder of this section we discuss existence of Nash equilibria in multi-player games. Theorem~\ref{thm:iac-ne} below is our most general result.

\begin{theorem}\label{thm:iac-ne}
Let $\langle A,C,O,d,v,(\prec_a)_{a\in A}\rangle$ be a game, let $\Gamma\subseteq\mathcal{P}(C^\omega)$, and assume the following.
\begin{enumerate}
\item\label{cond:iac-ne0} The $\prec_a$ are strict weak orders.
\item\label{cond:iac-ne1} The game $\langle C,D,W\rangle$ is determined for all $W\in\Gamma$, $D \subseteq C^*$.
\item\label{cond:iac-ne2} For every $a\in A$ and $\prec_a$-terminal interval $I$ and $\gamma\in C^*$, we have $\left (v^{-1}[I]\cap \gamma C^\omega \right ) \in\Gamma$.
\item\label{cond:iac-ne3} For every play $p\in C^\omega$ and increasing sequence $\varphi: \mathbb{N} \to \mathbb{N}$, if $d(p_{<\varphi(n)}) = a$ and $G_a(p_{<\varphi(n+1)}) \subsetneq G_a(p_{<\varphi(n)})$ for all $n\in\mathbb{N}$, then $v(p)\in \cap_{n\in\mathbb{N}}G_a(p_{<\varphi(n)})$.
\item\label{cond:iac-ne4} For all $\gamma\in C^\omega$, there exists $s$ such that $g_a(\gamma,s)=G_a(\gamma)$.
\end{enumerate}
Then the game $\langle A,C,O,d,v,(\prec_a)_{a\in A}\rangle$ has a Nash equilibrium.

\begin{proof}
Since the proof is similar to that of Theorem 2.9 in~\cite{leroux3}, we rephrase and give it a more intuitive flavour. Let $\sigma$ be a strategy profile where every player is using a witness to Lemma \ref{lem:max-guarant1}. Let $p$ be the induced play. We now turn $\sigma$ into a Nash equilibrium with $p$ as induced play by use of threats. More specifically, at each node $p_{<n}$ we let the players other than $a := d(p_{<n})$ threaten Player $a$ that if she deviates from $p$ exactly at $p_{<n}$, they will team up against her at every subsequent position $\gamma$ after $p_{<n}$ other than those extending the prescribed $p_{<n+1}$.

We claim that if they team up, they can prevent Player $a$ from getting better outcome than $v(p)$ by deviating to $\gamma$, which will suffice. Let us build a win-lose game $\langle C,D,W\rangle$, with Player $a$ against her threatening opponents gathered as a meta-player, and where the winning set for Player $a$ is defined by $W = v^{-1}[I] \cap \gamma C^\omega$, where $I := \{o\in O\,\mid\, v(p)\prec_a o\}$, and $D$ is defined by $D = d^{-1}(\{a\}) \cup \left (C^* \setminus \gamma C^*\right )$. This game is determined by Assumptions~\ref{cond:iac-ne1} and \ref{cond:iac-ne2}, and Player $a$ looses it, otherwise her winning strategy would guarantee that $v(p)\notin G_a(p_{<n})$ and thus contradict the choice of $p$. Therefore the threat of the opponents of Player $a$ is effective.
\end{proof}
\end{theorem}

%Instead of Condition~\ref{{thm:iac-ne3} in Theorem~\ref{thm:iac-ne} and prior Lemmas, \cite{paulyleroux2} used the following stronger condition: For every $p\in C^\omega$, sequence $(s_n)_{n\in\mathbb{N}}$ of strategies for $X\in\{a,b\}$, and increasing $\varphi: \mathbb{N} \to \mathbb{N}$, if $g_X(p_{<\varphi(n+1)},s_{n+1}) \subsetneq g_X(p_{<\varphi(n)},s_n)$ for all $n\in\mathbb{N}$, then $v(p)\in \cap_{n\in\mathbb{N}}g_X(p_{<\varphi(n)},s_n)$. This slight generalisation makes Corollary~\ref{cor:finite-graph-Borel}

Theorem~\ref{thm:NE-eq} below is a simpler version of Theorem~\ref{thm:iac-ne} that does only involve primitive notions from the definition of a game. Especially, it does not refer to the notion of guarantee. Via a necessary and sufficient condition, it shows how essential the feature of Example~\ref{example:nonash} is for the existence of Nash equilibrium.

\begin{theorem}\label{thm:NE-eq}
Let $g$ be a $\langle A,C,O,d,v,(\prec_a)_{a\in A}\rangle$ be a game where the $\prec_a$ are strict weak orders and $v$ is Borel-measurable. The following are equivalent.

\begin{enumerate}
\item For every $X\in A$ and $(p^n)_{n\in\mathbb{N}}$ sequence of plays in $C^\omega$ converging to some $p$, and increasing $\varphi: \mathbb{N} \to \mathbb{N}$, if for all $n\in\mathbb{N}$ we have $d(p_{<\varphi(n)}) = X$, $p_{<\varphi(n)} = p^n_{<\varphi(n)}$, $p_{\varphi(n)} \neq p^n_{\varphi(n)}$, and $v(p^n) \prec_X v(p^{n+1})$, then $v(p^n) \prec_X v(p)$ for all $n\in\mathbb{N}$.

\item Every finite-branching game derived from the original game by pruning has an NE.
\end{enumerate}

\begin{proof}
Let us first prove $1.\Rightarrow 2.$ by invoking Theorem~\ref{thm:iac-ne}. More specifically,  let $T$ be a finite-branching, infinite subtree of $C^\omega$ and consider the restriction of the original game to $T$. Conditions~\ref{cond:iac-ne2} and \ref{cond:iac-ne1} follow from Borel measurability and~\cite{martin}. Condition~\ref{cond:iac-ne4} comes from Lemma~\ref{lem:max-guarant0} (actually a straightforward extension of Lemma~\ref{lem:max-guarant0} to trees with finite-yet-unbounded branching), and Condition~\ref{cond:iac-ne3} follows directly from the assumption.

For $2.\Rightarrow 1.$, let $X\in\{a,b\}$ and let $(p^n)_{n\in \mathbb{N}} \to p\in C^\omega$ and increasing $\varphi: \mathbb{N} \to \mathbb{N}$ such that for all $n\in\mathbb{N}$ we have $d(p_{<\varphi(n)}) = X$, $p_{<\varphi(n)} = p^n_{<\varphi(n)}$, $p_{\varphi(n)} \neq p^n_{\varphi(n)}$, and $v(p^n) \prec_X v(p^{n+1})$. Let $T$ be the tree made of the prefixes of $p$ and the $p^n$. Since the game induced by $T$ has an NE and its tree structure is similar to Example~\ref{example:nonash}, $v(p^n) \prec_X v(p)$ must hold for all $n\in\mathbb{N}$.
\end{proof}
\end{theorem}

However, the modification of Example~\ref{example:nonash} below shows that the conditions of Theorem~\ref{thm:NE-eq} are not necessary for the mere existence of Nash equilibria.

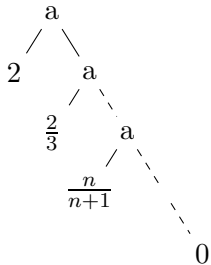
\begin{tikzpicture}[level distance=8mm] %%% recursive definition of a tree!
\node{a}[sibling distance=10mm]
	child{node{$2$}[sibling distance=10mm]}
    child{node{a}[sibling distance=10mm]
		child{node{$\frac{2}{3}$}}
		child{node{a} [sibling distance=10mm] edge from parent[dashed]
			child{node{$\frac{n}{n+1}$} edge from parent[solid]} %%% "solid" is supposed to be the default value, but it seems it has to be specified after using "dashed".
			child{node{} edge from parent[dashed]
				child{node{}edge from parent[draw=none]}
				child{node{$0$} edge from parent[dashed]}
			}
		}	
	};
\end{tikzpicture}

For further comparison, the preparatory work before \cite[Theorem 2.9]{leroux3}
considers strict well-orders only; then \cite[Theorem 2.9]{leroux3}
considers strict well-founded orders, since linear extensions of these
make it possible to invoke the special, linear case, knowing that any Nash
equilibrium for these extensions is still a Nash equilibrium for the
original preferences. However, let us explain why
Theorem~\ref{thm:iac-ne} assumes that preferences are strict weak orders,
instead of more general strict partial orders. In the preparatory work before both results, the algorithm that builds a play step by step makes decisions
based on the guarantees that the subgames offer. If the guarantees of one
player were not ordered by a strict weak order, the player might eventually
regret a previous decision, in the same way that backward induction on
partially ordered preferences may not yield a Nash equilibrium (see \textit{e.g.}, \cite{Krieger03} for a concrete example or page 3 of~\cite{SLR09} for a generic one). So the
algorithm has to run on strict weak orders. (In \cite[Theorem
2.9]{leroux3} it even runs on strict linear orders.)

If we wanted to consider strict partial orders and extend them linearly for the algorithm to work, we would potentially run into two problems: first, there may not exist any linear extension preserving Condition~\ref{cond:iac-ne3}. Second, assumptions~\ref{cond:iac-ne1} and \ref{cond:iac-ne2} of
Theorem~\ref{thm:iac-ne} make sure that the win-lose games associated
with the $\prec_a$-terminal intervals are determined, which is a requirement
for the proof to work. If the preferences were not strict weak orders, we
might think of replacing the condition on terminal intervals by a
condition on the upper-closed sets and then
extend the preferences linearly for the algorithm to work, but in the
special case where the preference of one player were the empty relation,
every subset would be an upper-closed set and its preimage by $v$ would be
in the pointclass with nice closure property, by assumption. If, in
addition, each outcome is mapped to
at most one play, it implies that each subset of $C^\omega$ is in the
pointclass, so Theorem~\ref{thm:iac-ne} could be used with the axiom of
determinacy only, but not with, \textit{e.g.}, Borel determinacy. On the
contrary, \cite[Theorem 2.9, Assumption 3]{leroux3} is not an issue since
there are only countably many outcomes in that setting.

Theorem \ref{thm:iac-ne} has a corollary pertaining to sequential games with real-valued payoffs. Rather than the usual Euclidean topology, we consider the lower topology generated by $\{(-\infty, a) \mid a \in \mathbb{Q}\}$. This space will be denoted by $\mathbb{R}_>$. Note that continuous functions with codomain $\mathbb{R}_>$ are often called upper semi-continuous. As $\id : \mathbb{R}_> \to \mathbb{R}$ is complete for the $\Sigma_2^0$-measurable functions \cite{weihrauchd,stein}, we see that the Borel sets\footnote{As $\mathbb{R}_>$ is not metric (but still countably based), the definition of the Borel hierarchy has to be modified as demonstrated by \name{Selivanov} \cite{selivanov3}. A move towards definitions of Borel measurability on even more general spaces can be found in \cite{paulydebrecht2}.} on $\mathbb{R}_>$ are the same as the Borel sets on $\mathbb{R}$. Moreover, if $(p^n)_{n \in \mathbb{N}}$ is a converging sequence of plays, and $P : \Cantor \to \mathbb{R}_>$ is a continuous payoff function, then $P(\lim_{i \in \mathbb{N}} p^i) \geq \limsup_{i \in \mathbb{N}} P(p^i)$. In particular, Condition~\ref{cond:iac-ne3} in Theorem \ref{thm:iac-ne} is always satisfied for the preferences obtained from upper semi-continuous payoff functions.

\begin{corollary}
\label{corr:uppercont}
Sequential games with countably many players, finitely many choices and upper semi-continuous payoff functions have Nash equilibria.
\end{corollary}

A rather simple argument allows us to transfer existence theorems for
equilibria in games with Borel-measurable valuations to Borel-measurable real-valued payoff functions with upper bound, if one is willing to
replace the original notions by their $\varepsilon$-counterparts. If $v :
\mathbf{S} \to (-\infty,0]^\omega$ is the Borel-measurable payoff function (with a component for each of the countably many players), then
for every positive real $\epsilon$ we define $v_{\epsilon} : \mathbf{S} \to
\mathbb{N}^\omega$ by $v_{\epsilon}^{-1}((i_k)_{k \in \mathbb{N}}) := v^{-1}\left
(]-(i_1+1)\epsilon, i_1\epsilon] \times ]-(i_2+1)\epsilon,i_2\epsilon] \times \ldots \right )$. Then any $v_{\epsilon}$ is again a Borel measurable
valuation (as a product of countably many intervals is $\Pi^0_2$). Furthermore, we define the preferences $\prec_n$ for the $n$-th
player by $(i_k)_{k \in \mathbb{N}} \prec_n (j_k)_{k \in \mathbb{N}}$ iff $i_n < j_n$. Now every Nash equilibrium of the resulting
game is a $\epsilon$-Nash equilibrium of the original game, and every
subgame perfect equilibrium of the resulting game is a
$\epsilon$-subgame perfect equilibrium of the original game.

\begin{corollary}\footnote{In his survey \cite{mertens}, \name{Mertens}
sketches an observation by himself and \name{Neyman} that one may use
Borel determinacy to directly obtain the special case of this result for finitely many players and bounded payoffs.}
\label{corr:borelhasenash}
Sequential games with countably many players and Borel-measurable payoff
functions with upper-bounds admit $\varepsilon$-Nash equilibria.
\begin{proof}
By combining the statement of Theorem~\ref{thm:NE-eq} with the argument above. We can invoke Theorem~\ref{thm:NE-eq} as the preferences $\prec_n$ do not have any infinite ascending chains at all.
\end{proof}
\end{corollary}

\begin{corollary}
\label{corr:esubgame}
A sequential two-player zero-sum game with Borel measurable payoffs has $\varepsilon$-subgame perfect equilibria.
\begin{proof}
By Corollary \ref{cor:qbd-spe} and the argument above.
\end{proof}
\end{corollary}

\section{On the existence of Pareto-optimal NE}
\label{sec:pareto}
In this section we investigate very general classes of games that guarantee existence of Nash equilibria, and such that there exists an NE that is Pareto-optimal among all the profiles of the game (not just Pareto-optimal among all Nash equilibria). In the following, we shall assume that any outcome is realizable to avoid unnecessary case-distinctions.

\begin{lemma}
\label{lemma:paretospecialcase}
Let $\Gamma$ be a determined pointclass. Then every infinite sequential two-player game with a $\Gamma$-measurable outcome function and preferences $y \prec_a x_1 \prec_a \ldots \prec_a x_n$ and $y \prec_b x_n \prec_b \ldots \prec_b x_1$, has a Pareto-optimal NE.
\begin{proof}
By assumption that every outcome is realizable, there is some path $p$ through the game yielding a payoff that is not $y$. For each vertex along this path, by determinacy either the opponent can enforce the outcome $y$, or the controller can enforce some upper interval. As long as the opponent can enforce $y$, he can force the controller to play along the chosen path by threaten punishment by $y$ for deviation. If we ever reach a vertex where the controller (w.l.o.g.~$a$) can enforce $\{x_1, \ldots, x_n\}$, there will be some minimal upper set $\{x_i, \ldots, x_n\}$ (from her perspective) that she can enforce. By determinacy, again, the opponent can enforce $\{y, x_1, \ldots, x_i\}$. We then let both players play their enforcing strategy from this node onwards.

The constructed partial strategies can be extended in an arbitrary way to yield a Nash equilibrium with another outcome than $y$, and these are all Pareto-optimal.
\end{proof}
\end{lemma}

\begin{theorem}\label{thm:extensive-Pareto}
We fix a non-empty set of players $A$ and  a non-empty set of outcomes $O$. Let $\Gamma$ be a determined pointclass closed under rescaling and union with clopens. Then the following are equivalent for a family $(\prec_a)_{a \in A}$ of linear preferences:
\begin{enumerate}
\item The inverse of the preferences are well-founded and $\forall a,b\in A,\forall x,y,z\in O,\neg(z \prec_a y \prec_a x \wedge x \prec_b z\prec_b y)$.
\item Every finite sequential game (built from $A$, $O$, $(\prec_a)_{a \in A}$) with three leaves has a Pareto-optimal NE.
\item Every infinite sequential game (built from $A$, $O$, $(\prec_a)_{a \in A}$) with a $\Gamma$-measurable outcome function has a Pareto-optimal NE.
\end{enumerate}
\begin{proof}
\begin{description}
\item[$3.\Rightarrow 2.$] Clear.

\item[$2.\Rightarrow 1.$] By contraposition, let us assume that $z \prec_a y \prec_a x$ and $ x \prec_b z\prec_b y$, and note that the game below has only one NE yielding outcome $z$.

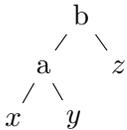
\begin{tikzpicture}[level distance=7mm]
\node{b}[sibling distance=10mm]
	child{node{a}[sibling distance=8mm]
		child{node{$x$}}
		child{node{$y$} }
	}
	child{node{$z$} };
\end{tikzpicture}

\item[$1. \Rightarrow 3.$] By \cite[Lemma 4]{leroux2016} the second assumption in $1.$ implies that there exists a partition $\{O_i\}_{i\in I}$ of $O$ and a linear order $<$ over $I$ such that $i<j$ implies $x<_ay$ for all $a\in A$ and $x\in O_i$ and $y\in O_j$, and such that $<_b\mid_{O_i}=<_a\mid_{O_i}$ or $<_b\mid_{O_i}=<_a\mid_{O_i}^{-1}$ for all $a,b\in A$. By the well-foundedness assumption, $I$ has a $<$-maximum $m$.

    Fix some $a\in A$, let $\{x_1,\dots,x_n\} := O_m$ (again, by well-foundedness, each slice is finite) such that $x_n <_a \dots <_a x_1$, let $A_0 := \{b\in A| <_b| _{O_m} = <_a|_{O_m}\}$, let $A_1 := A\backslash A_0$, and let $y\notin O_m$. Let us derive a new game on the same tree: each vertex of the original game owned by $b\in A$ is now owned by $A_0$ if $b\in A_0$ and by $A_1$ otherwise. Each play of the original game that induces an outcome outside of $O_m$ induces $y$ in the derived game. The new preferences are $y <_{A_0} x_n <_{A_0}\dots <_{A_0} x_1$ and $y <_{A_1} x_1 <_{A_1}\dots <_{A_1} x_n$. By Lemma \ref{lemma:paretospecialcase}, the derived game has a Pareto-optimal NE (which cannot yield $y$, as this is the only non-Pareto-optimal outcome). It is also a Pareto-optimal NE for the original game.
\end{description}
\end{proof}
\end{theorem}

The situation for non-linear orders is less clear. Certainly, whenever some linearization avoids the forbidden pattern from Theorem \ref{thm:extensive-Pareto} ($1.$), there will be a Pareto-optimal NE (as being Pareto-optimal w.r.t.~the linearization implies being Pareto-optimal w.r.t.~the original preferences). However, we do not know whether partial preferences such that any linearization has the forbidden pattern is enough to enable absence of Pareto-optimal NE. Two examples that could potentially play a similar role to the generic counterexample in Theorem \ref{thm:extensive-Pareto} ($2. \rightarrow 1$.) follow:

\begin{example}
We consider a finite two-player game with outcomes $\{x,y, z, \alpha,\beta,\gamma\}$, preferences
$\gamma \prec_a y \prec_a x$ and $z \prec_a \beta \prec_a \alpha$ and $x \prec_b z \prec_b y$ and $\alpha \prec_b \gamma \prec_b \beta$ and game tree:

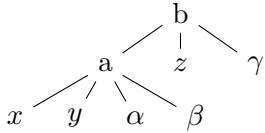
\begin{tikzpicture}[level distance=7mm]
\node{b}[sibling distance=10mm]
	child{node{a}[sibling distance=8mm]
		child{node{$x$}}
		child{node{$y$}}
		child{node{$\alpha$}}
		child{node{$\beta$}}
	}
	child{node{$z$}}
	child{node{$\gamma$}};
\end{tikzpicture}

The preferences avoid the forbidden pattern from Theorem \ref{thm:extensive-Pareto} ($1.$); but the pattern is present in any linear extension. In the Nash equilibria of the game, player $b$ is choosing either $z$ or $\gamma$; and player $a$ is choosing $x$ or $\alpha$. In particular, the potential equilibrium outcomes are $z$ and $\gamma$ -- precisely those outcomes that are not weakly Pareto-optimal (because every player prefers $y$ to $z$ and $\beta$ to $\gamma$). Both $y$ and $\beta$ would even have been Pareto-optimal.
\end{example}

\begin{example}
We consider a finite two-player game with outcomes $\{x,y,z,t\}$, preferences $t,z \prec_a x,y$ and $x \prec_b z \prec_b y \prec_b t$ and game tree:

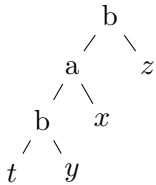
\begin{tikzpicture}[level distance=7mm]
\node{b}[sibling distance=10mm]
	child{node{a}[sibling distance=8mm]
		child{node{b}[sibling distance=8mm]
			child{node{$t$}}
			child{node{$y$}}
		}
		child{node{$x$}}
	}
	child{node{$z$} };
\end{tikzpicture}

The preferences are strict weak orders and avoid the forbidden pattern from Theorem \ref{thm:extensive-Pareto} ($1.$); but the pattern is present in any linear extension. The only equilibrium outcome is $z$, despite everyone preferring $y$ \ (\footnote{It may be an interesting remark that in this game every player would benefit, if $b$ could not choose $t$ at his second move.}).
\end{example}

\section{Absence of subgame perfect equilibria}
\label{sec:nosubgameperfect}
In this section we will show that in the simultaneous absence of continuity and the antagonistic/zero-sum property, even a two-player game with three distinct outcomes may fail to have subgame perfect equilibria. It is a straightforward consequence that moving to $\varepsilon$-subgame perfect equilibria cannot help, either. As our (counter-) Example \ref{example:nonash}, the valuation function here is $\Delta_2^0$-measurable, hence, in a sense, not very discontinuous. A similar counterexample is also exhibited in \cite[Example 3]{solan}.
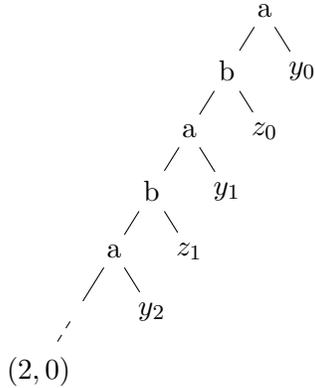
\begin{example}
\label{ex:nosubgameperfect}
The following game where $z <_a y <_a x$ and $x <_b z <_b y$ has no subgame
perfect equilibrium.

\begin{tikzpicture}[level distance=8mm]
\node{a}[sibling distance=10mm]
        child{node{b}[sibling distance=10mm]
                child{node{a}[sibling distance=10mm]
                        child{node{b}[sibling distance=10mm]
                                child{node{a}[sibling distance=10mm]
                                        child{node{}[sibling distance=10mm]
                                                child{node{$x$} edge from
parent[dashed]}
                                                child{node{} edge from
parent[white]}
                                        }
                                        child{node{$y$}}
                                }
                                child{node{$z$}}
                        }
                        child{node{$y$}}
                }
                child{node{$z$}}
        }
        child{node{$y$}};
\end{tikzpicture}

The game above is formally defined as
$\langle\{a,b\},\{0,1\},d,O,v,\{<_a,<_b\}\rangle$, where
$d^{-1}(a):=0^{2*}$ and $v(0^{\omega}) := x$ and
$v[0^{2*}1\{0,1\}^{\omega}]:=\{y\}$ and $v[0^{2*+1}1\{0,1\}^{\omega}]:=\{z\}$.
\begin{proof}
Assume for a contradiction that there is a subgame perfect equilibrium for this
game. Then no subprofile (starting at some node in $0^*$) induces the outcome
$x$, because Player $b$ could then switch to the right and obtain $z$. So for
infinitely many nodes in $0^*$, Players $a$ or $b$ chooses $1$. Also, if Player
$b$ chooses $1$ at some node $0^{2n+1}$, Player $a$ chooses $1$ at the node
$0^{2n}$ right above it. This implies that every subprofile rooted at nodes in
$0^{2*}$ induces the outcome $y$, and subsequently, Player $b$ always chooses
$0$ at nodes in $0^{2*+1}$. But then Player $a$ could always choose $0$ and
obtain $x$, contradiction.
\end{proof}
\end{example}

Conversely, consider a game with finitely many players and outcomes such that the pattern from Example~\ref{ex:nosubgameperfect} is absent from the preferences of any two players. If the outcome function is measurable in the Hausdorff difference hierarchy (of the open sets), \cite[Corollary 2]{leroux2016} says that the game has a Pareto-optimal subgame perfect equilibrium.

A further example shows us that we can rule out subgame perfect equilibria even with stronger conditions on the functions by using countably many distinct payoffs. This example no longer extends to $\varepsilon$-subgame perfect equilibria.
\begin{example}
\label{ex:nosubgameperfect2}
The following game where $y_n := (2^{-n},2^{-n})$ and $z_n := (0,2^{-n-2})$ for all $n\in\mathbb{N}$ has no subgame
perfect equilibrium, although the payoff functions are upper-semicontinuous.

\begin{tikzpicture}[level distance=8mm]
\node{a}[sibling distance=10mm]
        child{node{b}[sibling distance=10mm]
                child{node{a}[sibling distance=10mm]
                        child{node{b}[sibling distance=10mm]
                                child{node{a}[sibling distance=10mm]
                                        child{node{}[sibling distance=10mm]
                                                child{node{$(2,0)$} edge from
parent[dashed]}
                                                child{node{} edge from
parent[white]}
                                        }
                                        child{node{$y_2$}}
                                }
                                child{node{$z_1$}}
                        }
                        child{node{$y_1$}}
                }
                child{node{$z_0$}}
        }
        child{node{$y_0$}};
\end{tikzpicture}

\begin{proof}
If the payoffs are $(2,0)$, Player $b$ can improve her payoff as late as required, so there are infinitely many "right" choices in a subgame perfect equilibrium. If the payoffs are not $(2,0)$ then Player $a$ chooses "right" at every node that she owns, so that Player $b$ chooses "left", but then Player $a$ chooses "left" too.
\end{proof}
\end{example}

In \cite{solan2}, an intricate counterexample is provided showing that subgame-perfect equilibria may even fail to exist for probabilistic strategies (in a two-player game with Borel measurable payoff functions).

\section{Outlook}
There is one open question regarding infinite sequential games with real-valued payoff functions (see the question mark in the Overview table on page \pageref{fig:overview}), namely:
\begin{question}
Do games with countably many players and upper-semicontinuous payoff functions have $\varepsilon$-subgame perfect equilibria?
\end{question}

It seems surprising to have both Nash equilibria and $\varepsilon$-subgame perfect equilibria, but no subgame perfect equilibria -- but this is precisely the situation for finitely many players. On the other hand, given the split between finitely many players and countably many players for lower-semicontinuous payoff functions, one should be cautious about assuming that this result would extend. Thus, we do not present a conjecture regarding the answer to the open question.

The results in this paper are generally not constructive -- but neither is Nash's theorem in \cite{nash}, cf.~\cite{paulyincomputabilitynashequilibria}. The extent of non-constructivity is investigated in \cite{paulyleroux3-cie,paulyleroux3-arxiv}.

The condition on the payoff functions used in Section \ref{sec:ascendingpreferences} seems to merit further investigation. This was that for any sequence $(p^i)_{i \in \mathbb{N}}$ converging to $p$ in $C^\omega$, we find that $\forall i \in \mathbb{N} \ v(p^i) \prec v(p^{i+1})$ implies $\forall i \in \mathbb{N} \ v(p^i) \prec v(p)$. This is a weaker condition than continuity of the function where the upper order topology is used on the codomain, which still seems to be strong enough to formulate some results. In a sense, it is a
weakening of continuity that is orthogonal to Borel-measurability. As an example, a result by \name{Gregoriades} (reported in \cite{pauly-ordinals}) shows that any function of this type from Baire space to the countable ordinals has to be bounded.

In settings inspired by verification, sequential games often take place on a finite graph rather than an infinite tree. In this case, rather than seeking Nash equilibria built from arbitrary strategies, it is desirable to obtain Nash equilibria built from strategies realizable by finite automata. The high-level ideas of this paper can be translated (with some effort) into this setting, as demonstrated by the authors in \cite{paulyleroux4-arxiv}.

\bibliographystyle{eptcs}
\bibliography{../spieltheorie}

\section*{Acknowledgements}
We are very grateful to Vassilios Gregoriades and the anonymous referees of the conference version for their helpful comments. This work benefited from the Royal Society International Exchange Grant IE111233. The authors participated in the Marie Curie International Research Staff Exchange Scheme \emph{Computable
Analysis}, PIRSES-GA-2011- 294962, and are also supported by the ERC inVEST (279499) project.

\end{document}